\documentclass[letterpaper, 10 pt, conference]{ieeeconf}

\IEEEoverridecommandlockouts                

\usepackage{graphics} 
\usepackage{epsfig} 
\usepackage{mathptmx} 
\usepackage{times} 
\usepackage{amsmath}
\usepackage{xcolor}
\usepackage{amssymb}
\usepackage{subcaption}
\usepackage{multicol}
\usepackage[amssymb]{SIunits}

\newtheorem{proposition}{Proposition}

\newtheorem{definition}{Definition}
\newtheorem{remark}{Remark}
\newtheorem{lemma}{Lemma}

\newcommand{\norm}[1]{\vert\vert#1\vert\vert}
\DeclareMathOperator*{\argmin}{arg\,min}

\title{\LARGE \bf
Safe Autonomous Docking Maneuvers for a Floating Platform based on Input Sharing Control Barrier Functions}
\author{Akshit Saradagi$^*$, Avijit Banerjee, Sumeet Satpute and George Nikolakopoulos
\thanks{$^*$ Corresponding author}
\thanks{The authors are with the Robotics and Artificial Intelligence Subject at the Department of Computer Science, Electrical and Space Engineering, Lule\aa~University of Technology, Sweden. Emails: \tt{akssar@ltu.se; aviban@ltu.se; sumsat@ltu.se; geonik@ltu.se}}
}

\begin{document}

\maketitle
\thispagestyle{empty}
\pagestyle{empty}

\begin{abstract}
In this article, we present a control strategy for the problem of safe autonomous docking for a planar floating platform (Slider) that emulates the movement of a satellite. Employing the proposed strategy, Slider approaches a docking port with the right orientation, maintaining a safe distance, while always keeping a visual lock on the docking port throughout the docking maneuver. Control barrier functions are designed to impose the safety, direction of approach and visual locking constraints. Three control inputs of the Slider are shared among three barrier functions in enforcing the constraints. It is proved that the control inputs are shared in a conflict-free manner in rendering the sets defining safety and visual locking constraints forward invariant and in establishing finite-time convergence to the visual locking mode. The conflict-free input-sharing ensures the feasibility of a quadratic program that generates minimally-invasive corrections for a nominal controller, that is designed to track the docking port, so that the barrier constraints are respected throughout the docking maneuver. The efficacy of the proposed control design approach is validated through various simulations.
\end{abstract}
\begin{keywords}
Autonomous Docking, Control Barrier Functions, Safety, Visual Locking, Planar Floating Platform, Quadratic Programming, Control Applications.
\end{keywords}
\section{Introduction}
Recent technological advancements in space have significantly altered our perception of space activities. Space organizations around the globe are looking forward to autonomous robotic missions enabling critical space operations such as on-orbit service, visual servoing in the proximity of small celestial bodies, docking and active debri removal. 
Such ambitious and challenging applications necessitate development of efficient and effective autonomous guidance, navigation, and control (GNC) technologies \cite{sorgenfrei2014operational}. 
In order to support the complex technological progress and demonstrate and validate complex GNC strategies, hardware-in-loop spacecraft test-bed facilities have emerged as an economical alternative to the costlier in-space demonstrations \cite{elissa, testf2}. Most of these test-beds emulate spacecraft's orbital motion by a floating platform over a flat table where an air cushion allows the platform to levitate over the surface. 
%

In this article, we propose safe and autonomous docking strategies for the planar floating satellite platform (referred to as Slider in the rest of the article)~\cite{Slider_RAL, Docking_Mechanism} that has been designed to emulate zero-gravity motion of a spacecraft. This floating platform (Figure \ref{Ref_frame}) is supported by three air bearings, which release compressed air to form an air cushion that allows it to levitate over a smooth surface. The air cushion provides friction-less translation and rotational motion on a relatively flat surface. 
Such a floating platform, operating on the flat surface, provides a space-representative environment (though on the 2-D operational surface of the flat-table) to develop and evaluate advanced GNC algorithms.  More insights on the platform's design are documented in \cite{Slider_RAL, Docking_Mechanism}.
\begin{figure}
\centering
\centerline{\includegraphics [width=\columnwidth] {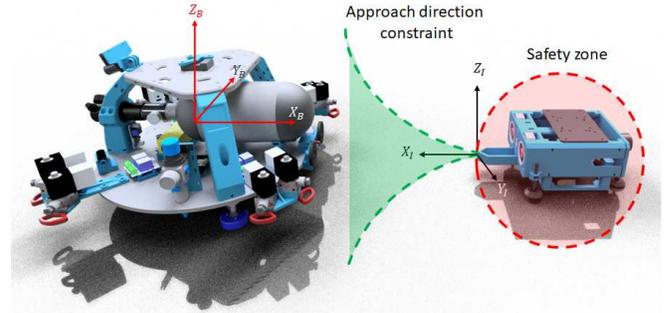}}
\caption{A schematic illustrating the frames of reference, the safety and the direction of approach constraints involved in the scenario where the planar floating platform, Slider (on the left), docks onto a docking port (on the right). }\label{Ref_frame}
\end{figure}

In Figure \ref{Ref_frame}, we illustrate an autonomous and safe docking scenario that is being investigated in this work. The aim is to design a control strategy for the Slider platform that achieves docking, while respecting the safety and direction of approach constraints (necessary for smooth docking). A layer of autonomy is added to the operation by requiring that a visual lock be maintained on the docking port throughout the maneuver. 

In the past decade, the control barrier functions (CBF) approach has emerged as an elegant framework for enforcing state constraints in control systems \cite{barrier_main_ref, main_2}. In this technique, barrier functions are constructed in such a way that their super-levels capture subsets of the state-space where the constraints are satisfied and the intersection of these sets ($\mathcal{H}$) is rendered forward invariant and asymptotically stable, that is, $x(t_0) \in \mathcal{H}$ implies $x(t) \in \mathcal{H}$ for all $t\geq t_0$ and $x(t_0) \notin \mathcal{H}$ implies $x(t) \rightarrow \mathcal{H}$ as $t \rightarrow \infty$. Methods have been devised for scenarios where the constraints are time-varying \cite{time_varying_human_assist, time_varying_ttl} and where attractivity to the safe set is ensured in finite time \cite{finite-time-tbd}. Minimum-norm quadratic programs are then designed to enforce the barrier constraints, while minimally deviating from a nominal controller. In this article, we present a multi-CBF based solution to the safe and autonomous docking problem for the Slider platform.

\textbf{Statement of contributions}: In contrast to the previously established state of the art, the contributions of this article are as follows. 
(a) A Control Barrier Functions based provably safe autonomous docking mechanism is presented for the Slider, that takes into account the safety aspects around the docking station and direction of approach constraints. A layer of autonomy is encoded into the design by imposing the constraint of visual lock on the docking port. This scenario presents us with an opportunity to analyse the unexplored and challenging scenario where multiple inputs are shared among multiple barrier functions. (b) The results from literature concerning single-input sharing among multiple barriers, presented in \cite{Sharing_Journal, Sharing_Automation} and multi-input sharing among barrier functions with non-intersecting boundaries in \cite{Multiple-barriers}, cannot be employed for the scenario in this article. We present theoretical proof that three control inputs of the Slider are shared harmoniously among three barriers in rendering the sets defining safety and visual locking forward invariant. (c) For the scenario where the Slider is initialized in the safe zone but with no visual lock on the docking port, we prove the asymptotic convergence to the visual locking mode and then present a control strategy to achieve finite-time convergence to the visual locking mode.


The rest of the article is organised as follows. In Section \ref{Sec:Notations}, we define the notations used in this work and recall important results from the theory of control barrier functions. In Section \ref{Sec:Formulation}, we present the dynamics of the Slider platform, derive the barrier functions capturing safety and visual locking and define our problem statement. In Section \ref{CBFs},
we present a two-loop control architecture used for the Slider and derive the control barrier functions that help in establishing safety and visual locking. In Section \ref{Sec:input_sharing}, we define the control sharing property among multiple functions, analyze the feasibility of control sharing and show that the CBFs harmoniously share the control inputs in establishing invariance of the safe sets that ensure safety and visual locking throughout the docking maneuver. We present our simulation results in Section \ref{Sec:Simulation_Results} followed by concluding remarks in Section \ref{Sec:Conclusions}.

\section{Notations and Preliminaries}\label{Sec:Notations}
We denote the set of real numbers by $\mathbb{R}$ and by
$\mathbb{S}^1$, we denote the unit circle. The notation $L_{f} h(x)$, is used to denote the Lie derivative of a continuously differentiable function $h(x)$ along the vector field $f(x)$, that is, $\frac{\partial h(x)}{\partial x} f(x)$. The boundary of a set $\mathcal{S}$ is denoted as $\partial\mathcal{S}$. A continuous function $\alpha: (-b,a)\rightarrow (-\infty,\infty)$ is said to be an extended class-${\cal K}$ function if $\alpha(0)=0$ and it is strictly increasing.  

Next, we recall some results concerning control barrier functions (CBFs) from \cite{barrier_main_ref, main_2}, which will be used in this work to design a control strategy for the Slider platform that ensures safety and visual locking throughout the docking maneuver. Consider a control-affine dynamical system 
\begin{equation}
    \dot{x}=f(x)+g(x)u, \;\;\; x \in \mathcal{X}\subset \mathbb{R}^n, u \in \mathcal{U} \subset \mathbb{R}^m  
    \label{affine}
\end{equation}
where $f$ and $g$ are Lipschitz continuous functions. Let $\mathcal{S}\subset\mathcal{X}$ be the region of the state-space, which is deemed as a safe region for the operation of \eqref{affine}. Let $h(x): \mathcal{D}\subset\mathcal{X}\rightarrow \mathbb{R}$ be a continuously differentiable function, with $\mathcal{S}\subset\mathcal{D}$, such that $\mathcal{S}:=\{x \in$ $\left.\mathcal{X} \mid h(x) \geq 0\right\}$, that is, $\mathcal{S}$ is a zero super-level set of the function $h(x)$. The set $\mathcal{S}$ is rendered safe if the control input to \eqref{affine} ensures positive invariance of the set $\mathcal{S}$, that is, $x(t_0) \in \mathcal{S}$ implies $x(t) \in \mathcal{S}$ for all $t\geq t_0$. In addition, if $\mathcal{S}$ is rendered asymptotically stable, when the system is initialized in $\mathcal{D}\setminus\mathcal{S}$, a measure of robustness can be incorporated into the notion of safety.

The following definition introduces the notion of a control barrier function and presents a condition, the verification of which ensures the safety of the dynamical system \eqref{affine}.
\begin{definition}[A control barrier function \cite{barrier_main_ref}] \label{Barrier_Defn}
A continuously differentiable function $h(x): \mathcal{D} \rightarrow \mathbb{R}$ is a control barrier function, if there exists a real parameter $\gamma>0$ and a generalized class-$\mathcal{K}$ function $\alpha$, such that for all $x \in \mathcal{D}$,
\begin{equation}
\sup _{u \in \mathcal{U}}\left\{L_{f} h(x)+L_{g} h(x) u+\gamma \alpha(h(x))\right\} \geq 0.
\label{Barrier_Condition}
\end{equation}
\end{definition}

The forward invariance of $\mathcal{S}$ ($\dot{h}\geq0$ on $\partial\mathcal{S}$) and asymptotic stability of $\mathcal{S}$ ($\dot{h}>0$ in $\mathcal{D}\setminus\mathcal{S}$) are captured together in the condition \eqref{Barrier_Condition}. The existence of at least one input from the admissible set $\mathcal{U}$, which renders \eqref{Barrier_Condition} feasible, enables the design of a control input that ensures safety of the dynamical system \eqref{affine} under certain conditions. This is captured in the following proposition.
\begin{proposition}[\cite{barrier_main_ref}] \label{main_prop}
Let $\mathcal{S}\subset\mathcal{D}$ be the zero super-level set of a continuously differentiable function $h : \mathcal{D}\subset\mathcal{X}\rightarrow \mathbb{R}$. If $h$ is a control barrier function on $\mathcal{D}$ and $\frac{\partial h(x)}{\partial x}\neq 0$ on $\partial\mathcal{C}$, then any Lipschitz continuous controller $u(x) \in K_h(x)$, where \begin{equation}
K_h(x)=\{u \in \mathcal{U} \mid L_{f} h(x)+L_{g}h(x) u+\gamma \alpha(h(x)) \geq 0\}
\label{control_family}
\end{equation}
for the system \eqref{affine} renders $\mathcal{S}$ forward invariant and asymptotically stable in $\mathcal{D}$.
\end{proposition}
\section{Problem Formulation}\label{Sec:Formulation}
In this section, we present the dynamics of the Slider platform, derive the barrier functions that capture the safety and visual locking constraints and define our problem statement.
\subsection{Dynamics of the Slider platform}
The design of the planar floating satellite platform and the coordinate frames used to describe its equations of motion are depicted in Figure \ref{Ref_frame}. The inertial frame of reference ($\mathbb{I}=\{x_I,y_I,z_I\}$) is attached at the end of the docking port, which is assumed to be fixed on a flat table.
The body-fixed frame ($\mathbb{B}=\{x_B,y_B,z_B\}$) is attached to the Slider's Centre of Mass (CoM). 
The equations governing the motion of the Slider platform (derived in \cite{Slider_RAL}) in the state-space form are:
\begin{equation}\label{Eq:Slider_dynamics}
   \left[ \begin{matrix}
   {\dot{r_x}}  \\
   {\dot{r_y}}  \\
   {\dot{\theta }}  \\
   {{{\dot{v}}}_{x}}  \\
   {{{\dot{v}}}_{y}}  \\
   {\dot{\omega}_z}  \\
 \end{matrix} \right]=\left[ \begin{matrix}
   {{v}_{x}}\cos \theta -{{v}_{y}}\sin \theta   \\
   {{v}_{x}}\sin \theta +{{v}_{y}}\cos \theta   \\
   \omega_z   \\
   \omega_z {{v}_{y}}+\frac{1}{m}{f}_{x}  \\
   -\omega_z {{v}_{x}}+\frac{1}{m}{f}_{y}  \\
   \frac{\tau_z }{{{I}_{zz}}}  \\
\end{matrix} \right]
 \end{equation}
 where, $r_x,r_y \in\mathbb{R}$ describe the position of the Slider with respect to the inertial frame $\mathbb{I}$. The angle $\theta\in\mathbb{S}^1$ represents the orientation of the Slider with respect to the inertial frame. As various sensors and actuators are attached to the Slider's body, we prefer to consider the linear velocities $v_x, v_y \in \mathbb{R}$, directed along $x_B$ and $y_B$ and the rotational velocity $\omega_z$ directed along $z_B$, in the body-frame $\mathbb{B}$. The state of the Slider $(r_x,r_y,\theta,v_x,v_y,\omega_z)$ evolves in $\mathcal{Q}_1=\mathbb{R}^5\times \mathbb{S}^1$. The torque applied along $z_B$ is denoted by $\tau_z\in\mathbb{R}$ and $I_{zz}\in \mathbb{R}$ denotes the principal moment of inertia along the $z_B$ direction. The platform has been designed such that the other components of the moment of inertia matrix are negligible. The mass of the platform is denoted by $m$ and the forces along $x_B$ and $y_B$ axes of the body frame are denoted by $f_x$ and $f_y$ respectively. The Coriolis effects influencing the motion of the Slider feature in the dynamics through the terms $\omega_z {{v}_{y}}$ and $-\omega_z{{v}_{x}}$.
\subsection{Safe autonomous docking with visual locking}
The docking maneuver proposed in this work comprises of three key features. We present their formulations and elaborate on the necessity of incorporating these features into the design of the docking maneuvers in the following subsections.
\subsubsection{Tracking the docking port}
The port onto which the Slider platform is expected to dock is fixed on a flat table and the inertial frame of reference $\mathbb{I}$ is fixed to the tip of the docking port (Figure \ref{Ref_frame}). The docking port is oriented along the axis $x_I$ of the inertial frame. 
The Slider is expected to approach the docking port asymptotically with the opposite orientation, that is  $\theta(t) \rightarrow \pi$ as $t \rightarrow \infty$. If the Slider blindly tracks the docking port, there is a possibility of the Slider, on its way to the docking port, crashing into the docking station, which the docking port is part of. This necessitates the defining of unsafe zones around the docking port and the imposition of constraints on the direction of approach to the docking port. 
\subsubsection{Safe approach to the docking port} \label{Sub:Safety}
It is critical that the Slider approaches the docking port, while maintaining a safe distance from the docking station. We identify a region around the docking station (indicated in red, in Figure \ref{safety}) into which the Slider is prohibited from entering. It is also expected that the Slider approaches the docking port through a tapering funnel as illustrated in Figure \ref{safety} and identified in green. This is to enable smooth alignment of the ports as the Slider nears the docking port. 

A Cardioid centered at the docking port (origin of the inertial frame) and oriented along the axis $x_I$ is plotted in Figure \ref{cardioid_plot}. At the origin, the slopes of the symmetric curves of the Cardioid above and below the $x_I$-axis are zero. This cusp at the origin encodes the tapering funnel region, indicated in green in Figure \ref{safety}. The docking station (shown in red in Figure \ref{cardioid_plot}) is enclosed by the Cardioid and the safety constraint is encoded by defining as safe, the region outside the Cardioid
\begin{equation}
(x^2+y^2)^2+4ax(x^2+y^2)-4a^2y^2 \geq 0
\label{outside_cardioid}
\end{equation}
where $a$ is a positive real parameter that decides the size of the Cardioid. If it is ensured that the Slider remains outside the Cardioid throughout the docking maneuver, we indirectly ensure, in one shot, that both the safety and the direction of approach constraints are satisfied. 
 \begin{figure}
 \centering
 \subfloat[The area around the docking station, indicated in red, is considered unsafe. The area inside the tapering funnel, indicated in green, is considered a safe region to approach the docking port. \label{safety}]{\includegraphics[scale=0.85]{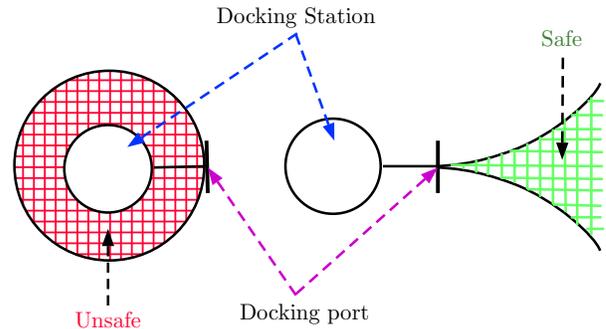}}
 \vspace{0.1cm}
 \subfloat[A Cardioid as a barrier function \label{cardioid_plot}]{\includegraphics[scale=0.55]{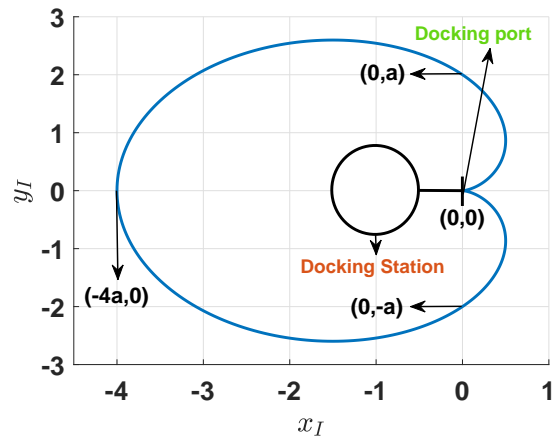}}
 \caption{Both the safety and the direction of approach constraints are captured in one shot, by considering the area outside (inside) the Cardioid to be safe (unsafe).}
 \end{figure}
\subsubsection{Visual lock on the docking port} \label{Sub:Visual_lock}
By imposing the constraint that once the docking port comes into the visual range of the Slider, it must be contained within it throughout the docking maneuver, we add a layer of autonomy to the mission. We term this constraint as visual locking. This enables the Slider to accomplish the docking maneuver using on-board sensors  and with minimal interaction with an external sensing infrastructure. Here we assume that a vision sensor on-board the Slider offers a range of $\theta_v$ around the orientation of the Slider. As illustrated in Figure \ref{fig:vis_lock}, the visual locking constraint is imposed by requiring that the line joining the Slider and the origin of the inertial frame (fixed at the docking port) be within the zone $[\theta_s-\theta_v, \theta_s+\theta_v]$. This leads us to the constraint 
\begin{equation}
    \tan^{-1}\left(\frac{-r_y}{-r_x}\right)\in[\theta_s-\theta_v, \theta_s+\theta_v].
    \label{Lock_main}
\end{equation}
\begin{figure}
    \centering
    \includegraphics[width=0.7\columnwidth]{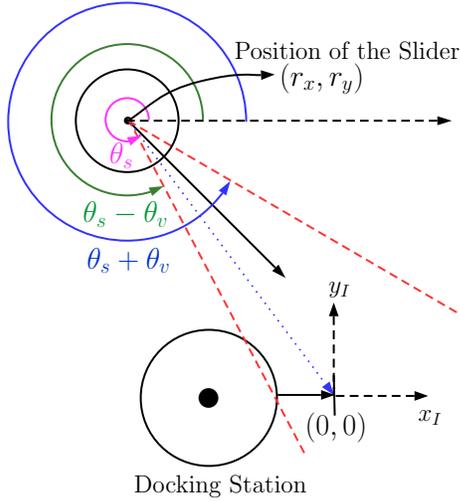}
    \caption{Illustration of the visual locking constraint. The blue dotted line connecting the Slider and the docking port must always be maintained in the visual zone contained within the dashed red lines.}
    \label{fig:vis_lock}
\end{figure}
\textbf{Problem Statement}: Design a feedback control strategy ${f_x, f_y, \tau_z}$ for the Slider dynamics in \eqref{Eq:Slider_dynamics}, to achieve safe autonomous docking with visual locking, as defined in Subsections \ref{Sub:Safety} and \ref{Sub:Visual_lock}, given the safety parameter $a$ defining the Cardioid and the parameter $\theta_v$ defining the range of visual sensing.
\section{Control barrier functions for safety and visual locking}\label{CBFs}
In this Section, we present a two-loop control architecture for the Slider and derive control barrier functions satisfying conditions in Definition \ref{Barrier_Defn}, that help in establishing the two constraints introduced in the previous section. 
\subsection{Control Architecture}
We employ a two-loop control architecture for the Slider. The inner-loop has access to the inputs (forces ${f_x, f_y}$ and torque $ \tau_z$) of the Slider and it is designed to track the velocities $\{v_x^d, v_y^d, \omega_z^d\}$ commanded by the outer-loop. Feedback linearization is first used to negate the nonlinearities in the dynamics \eqref{Eq:Slider_dynamics} and then a linear controller is designed to track the velocities commanded by the inner-loop. The controller  
\begin{equation}
\begin{aligned}
    f_x&=m(-v_y\omega_z+c_1(v_x-v_x^d)+\dot{v}_x^d) \\
    f_y&=m(v_x\omega_z+c_2(v_y-v_y^d)+\dot{v}_y^d)\\
    \tau_z&=I_{zz}(c_3(\omega_{z}-\omega_{z}^d)+\dot{\omega}_{z}^d)
    \end{aligned}
    \label{outer-loop}
\end{equation}
drives the errors $e_{v_x}=(v_x-v_x^d)$, $e_{v_y}=(v_y-v_y^d)$ and $e_{\omega_z}=(\omega_z-\omega_z^d)$ to zero exponentially and the positive constants $c_1,  c_2$ and $c_3$ are chosen to achieve the desired rates of convergence. The vector $\{\bar{v}_x,\bar{v}_y,\omega_z\}$ is treated as input to the outer-loop 
\begin{equation}\label{Eq:Slider_reduce}
   \left[ \begin{matrix}
   {\dot{r_x}}  \\
   {\dot{r_y}}  \\
   {\dot{\theta }}
 \end{matrix} \right]=\left[ \begin{matrix}
   \bar{v}_x   \\
   \bar{v}_y   \\
   \omega_z
\end{matrix} \right], \;\;\; \left[ \begin{matrix}
   {\bar{v}_x}  \\
   {\bar{v}_y}
 \end{matrix} \right]=\left[ \begin{matrix}
   {{v}_{x}}\cos \theta -{{v}_{y}}\sin \theta   \\
   {{v}_{x}}\sin \theta +{{v}_{y}}\cos \theta
\end{matrix} \right]
 \end{equation}
where $(r_x, r_y, \theta) \in \mathcal{Q}_2=\mathbb{R}^2\times \mathbb{S}^1$. We derive our control strategy for safe autonomous docking through these inputs. The velocity commands $(v_x^d, v_y^d)$ are then computed using the invertible transformation relating $\{\bar{v}_x,\bar{v}_y\}$ and $\{v_x,v_y\}$ in \eqref{Eq:Slider_reduce}. Under the assumption that the magnitudes of the linear and angular velocities that can be faithfully tracked by the inner-loop are bounded by $a>0$ and $b>0$ respectively, the admissible control input set is set to $\mathcal{U}=[-a \;\; a]\times[-a \;\; a]\times[-b \;\; b]\subset\mathbb{R}^3$. 

Next, we define three continuously differentiable functions, whose zero super-level sets define the subsets of the state-space where the constraints of safety and visual locking are satisfied. We then show that these functions are control barrier functions and satisfy the conditions in Definition \ref{Barrier_Defn}. 
\subsection{Barrier functions for visual locking}
From the visual locking condition \eqref{Lock_main}, we derive the following barrier functions
\begin{align}
    h_2(x)&=\tan^{-1}\left(\frac{-r_y}{-r_x}\right)-(\theta_s-\theta_v) \label{barrier_2}\\
    h_3(x)&=(\theta_s+\theta_v)-\tan^{-1}\left(\frac{-r_y}{-r_x}\right) \label{barrier_3}
\end{align}
If the state of the Slider $x=[r_x,r_y,\theta]'\in \mathcal{Q}_2  $ remains in the zero super level sets $\mathcal{V}_2=\{x \in \mathcal{Q}_2 \mid h_2(x) \geq 0\}$ and $\mathcal{V}_3=\{x \in \mathcal{Q}_2 \mid h_3(x) \geq 0\}$ for all $t \geq 0$, then a visual lock on the docking port, as formulated in \eqref{Lock_main}, is maintained throughout the docking maneuver.  The following Lemma presents verification that $h_2(x)$ and $h_3(x)$ are control barrier functions as in Definition~\ref{Barrier_Defn} on a sets $\mathcal{D}_2=\{x \in \mathcal{Q}_2 \mid h_2(x)\in(-b_1,\infty), b_1>0\}\supset\mathcal{V}_2$ and $\mathcal{D}_3=\{x \in \mathcal{Q}_2 \mid h_3(x)\in(-b_2,\infty), b_2>0\}\supset\mathcal{V}_3$ respectively.

Note that the functions $h_2$ and $h_3$ are continuously differentiable functions except in $\mathcal{R}=\{0,0\}\times\mathbb{S}^1$ ($r_x=r_y=0$), where the functions are undefined. This is expected and in practice not a concern, as the Slider is perfectly docked at this position and the visual locking condition need not be considered.
\begin{lemma}
\label{Lemma2}
The continuously differentiable functions $h_2(x)$ and $h_3(x)$ are control barrier functions in $\mathcal{Q}_2\setminus\mathcal{R}$, for the dynamics \eqref{Eq:Slider_reduce}. The set of control inputs $K_{h_2}(x)$ and $K_{h_3}(x)$ as defined in \eqref{control_family} are non-empty for all $x \in \mathcal{D}_2$ and $x \in \mathcal{D}_3$ respectively and any Lipschitz continuous controller $u(x)$ in $K_{h_2}(x)$ and $K_{h_3}(x)$ render the sets $\mathcal{V}_2$ and $\mathcal{V}_3$ forward invariant and asymptotically stable in $\mathcal{D}_2$ and $\mathcal{D}_3$ respectively.
\end{lemma}
\begin{proof}
In showing that $h_2(x)$ and $h_3(x)$ are CBFs on $\mathcal{D}_2$ and $\mathcal{D}_3$ respectively, the existence of the constants $k_2, k_3 > 0$ and extended class$-\mathcal{K}$ functions $\alpha_2$ and $\alpha_3$ such that 
\begin{align}
\sup _{u \in \mathcal{U}}\{b_{11}v_x+b_{12}v_y+b_{13}\omega_z\} & \geq -k_2 \alpha_2(h_2) \label{barrier_condition_2}\\
\sup _{u \in \mathcal{U}}\{c_{11}v_x+c_{12}v_y+c_{13}\omega_z\} & \geq -k_3 \alpha_3(h_3) \label{barrier_condition_3}
\end{align}
where 
\begin{align*}
    b_{11} &= \frac{\partial h_2}{\partial r_x}=-\frac{r_y}{(r_x^2+r_y^2)}& \;\;  c_{11} = \frac{\partial h_3}{\partial r_x}=\frac{r_y}{(r_x^2+r_y^2)} \\
    b_{12} &= \frac{\partial h_2}{\partial r_y}=\frac{r_x}{(r_x^2+r_y^2)}& \;\; c_{12} = \frac{\partial h_3}{\partial r_y}=-\frac{r_x}{(r_x^2+r_y^2)} \\
    b_{13} &= \frac{\partial h_2}{\partial \theta}=-1 \;\;& c_{13} = \frac{\partial h_3}{\partial \theta}=1
\end{align*}
are feasible for all $x$ in $\mathcal{D}_2$ and $\mathcal{D}_3$ respectively, must be shown. The equality in the set definition 
\begin{equation}
K_{h_2}(x)=\{u \in \mathcal{U} \mid b_{11}v_x+b_{12}v_y+b_{13}\omega_z  \geq -k_2 \alpha_2(h_2)\}
\label{kh}
\end{equation}
defines a hyper-plane in $\mathbb{R}^3$. For the subset of $\mathcal{D}_2$, where $h_2(x)>0$, the hyper-plane is offset from the origin of $\mathbb{R}^3$ such that a ball of non-zero radius centered at the origin is contained in $K_{h_2}(x)$. For the subset of $\mathcal{D}_2$, where $h_2(x)<0$, the hyper-plane is offset such that the origin of $\mathbb{R}^3$ is not contained in $K_{h_2}(x)$. Now, consider the length of the vector $(b_{11},\; b_{12},\; b_{13})$,  $\bar{b}=\sqrt{b_{11}^2+b_{12}^2+b_{13}^2}$. It is easy to see that $\text{inf} (\bar{b})=1$. Dividing both sides of the inequality in \eqref{kh} by $\bar{b}$, we have
\begin{equation}
\frac{b_{11}}{\bar{b}}v_x+\frac{b_{12}}{\bar{b}}v_y+\frac{b_{13}}{\bar{b}}\omega_z \geq -k_2 \alpha_2(h_2) \geq \frac{-k_2 \alpha_2(h_2)}{\bar{b}}.
\label{crucial}
\end{equation}
By choosing the constant $k_2>0$, the extended class$-\mathcal{K}$ function $\alpha_2$ and the constant $b_1>0$ that parameterizes $\mathcal{D}_2$ such that $\min(a,b)>k_2\alpha_2(b_1)$, the feasibility of \eqref{barrier_condition_2} is ensured in the set $\mathcal{D}_2$ due to the non-empty intersection between $K_{h_2}(x)$ and a ball of radius $\min(a,b)$ contained in the admissible input set $\mathcal{U}$. Therefore $h_2(x)$ is a CBF. By Proposition \ref{main_prop}, any Lipschitz continuous controller chosen from $K_{h_2}(x)$ establishes the forward invariance and asymptotic stability of $\mathcal{V}_2$ in $\mathcal{D}_2$.

We omit the proof that $h_3(x)$ is a CBF, as it follows along the same lines as the proof for $h_2(x)$.  
\end{proof}
\subsection{A Cardioid as a barrier function for safety}
To impose the safety and direction of approach constraints in the docking maneuver in one shot, we made a case in Subsection \ref{Sub:Safety}, for the use of the region outside the Cardioid $\mathcal{H}=\{x \in \mathcal{Q}_2 \;|\; h_1(x) \geq 0\}$, where
\begin{align}
    h_1(x)=(r_x^2+r_y^2)^2+4ar_x(r_x^2+r_y^2)-4a^2r_y^2
    \label{barrier_1}
\end{align}
as the safe set. In showing that $h_1(x)$ is a control barrier function as in Definition \ref{Barrier_Defn}, the existence of the constants $k_1>0$ and class$-\mathcal{K}$ function $\alpha_1$, such that 
\begin{equation}
\begin{aligned}
\sup _{u \in \mathcal{U}}\{L_{f} h(x)+L_{g} h(x) u\}&=\sup _{u \in \mathcal{U}}\{a_{11}v_x+a_{12}v_y+a_{13}\omega_z\} \\
& \geq -k_1 \alpha_1(h_1(x))
\label{barrier_condition_1}
\end{aligned}
\end{equation} 
where
\begin{equation}
\begin{aligned}
    a_{11} &= \frac{\partial h_1}{\partial r_x}=-(4(r_x^2+r_y^2)x+4a(r_x^2+r_y^2)+8ar_x^2) \\
    a_{12} &= \frac{\partial h_1}{\partial r_y} = -(4(r_x^2+r_y^2)r_y+8ar_xr_y-8a^2r_y) \\
    a_{13} &= \frac{\partial h_1}{\partial \theta}=0
\end{aligned}
\label{coeffs}
\end{equation}
is feasible in a set $\mathcal{D}_1 \supset\mathcal{H}$ must be shown. Note that the gradient of $h_1(x)$ defined by the state-dependent parameters $a_{1j}$ in \eqref{coeffs} vanishes at the origin of the $(r_x, r_y)$ plane and the parameters assume arbitrarily small values about the boundary $\partial\mathcal{H}$. The arguments made through the equation \eqref{crucial} in the proof of Lemma \ref{Lemma2} are not valid for the function $h_1(x)$ in the presence of the input constraints defined through the set $\mathcal{U}$. However, the feasibility of the condition \eqref{barrier_condition_1} can be established for all $x \in \mathcal{H}$, where $h_1(x)\geq0$. The consequences of this fact with regards to invariance of the safe set are captured in the following Lemma.
\begin{lemma}
\label{Lemma1}
The set of control inputs $K_{h_1}(x)$ as defined in \eqref{control_family} is non-empty for all $x \in \mathcal{H}$ and any Lipschitz continuous controller $u(x)\in K_{h_1}(x)$ renders the set $\mathcal{H}$ forward invariant under the dynamics \eqref{Eq:Slider_reduce}.
\end{lemma}
\begin{proof}
The feasibility of \eqref{barrier_condition_1} in the set $\mathcal{H}$ can be established using the arguments in the proof of Lemma \ref{Lemma2} and it can be concluded that $K_{h_1}(x)$ is non-empty for any $k>0$ and any class-$\mathcal{K}$ function $\alpha_1(x)$. Through Proposition \ref{main_prop}, we conclude that any Lipschitz continuous controller $u(x)\in K_{h_1}(x)$ renders $\mathcal{H}$ forward invariant.
\end{proof}

\begin{remark}
The functions \eqref{barrier_2}, \eqref{barrier_3} and \eqref{barrier_1} are defined in terms of the variables $(r_x,r_y,\theta)$. The functions have relative degree one, considering $(v_x,v_y,\omega_z)$ as the inputs and relative degree two, if $(f_x, f_y , \tau_z)$ are considered as inputs. In the multi-barrier and multi-input scenario, establishing that the functions are control barrier functions and verification of the control sharing property (to be presented in section \ref{Sec:input_sharing}) becomes progressively challenging as the relative degree increases. In this work, by considering $(v_x,v_y,\omega_z)$ as the inputs, the functions $h_i$ have relative degree one and our analysis is significantly eased. Moreover, a high-performance inner-loop is designed to set $(f_x, f_y , \tau_z)$ such that the desired velocities commanded by the outer-loop are tracked. 
\end{remark}
\subsection{Quadratic programs for minimally invasive safety guarantees}\label{Quad_prog}
In practice, the CBF formulation is amenable to efficient online computation of feasible control inputs that ensure safety. As the functions $h_i(x)$ are CBFs, for a fixed $x$, each of the constraints \eqref{barrier_condition_1}, \eqref{barrier_condition_2} and \eqref{barrier_condition_3} are feasible linear constraints and we set up a minimum-norm Quadratic program 
\begin{equation}
    \begin{aligned}
u^{*} & =\argmin_{u \in \mathcal{U}}  \norm{u-u_{\text{nom}}(x)}^2 \\
\text{subject to} \ &:\  L_{f} h_i(x)+L_{g}h_i(x) u \geq -k_i\alpha(h_i(x)),\; i\in\{1,2,3\}
\end{aligned}
\label{opt-1}
\end{equation}
to determine a minimally invasive correction to a nominal controller $u_{\text{nom}}$ that ensures safety. The nominal controller $u_{\text{nom}}$ is designed to achieve a control objective with no consideration for the safety constraints. The classical Quadratic program \eqref{opt-1} with the quadratic cost and linear constraints in the decision variable $u=\{v_x,v_y,\omega_z\}$ can be efficiently solved at every sampling instant. In this work, we use a linear controller $u_{\text{nom}}(x)=[-p_{1}r_x \;\; -p_{2}r_y \;\; -p_{3}(\theta-\pi)]^{\top}$, $p_i>0$ as the nominal controller for the Slider kinematics \eqref{Eq:Slider_reduce}, to exponentially drive the Slider to the origin of the $(x, y)$ plane with its orientation opposite to that of the orientation of the docking port. More sophisticated controllers, like the Model Predictive Controller (MPC) incorporating input and state constraints can also be chosen as a nominal controller. Establishing the Lipschitz continuity of $u^{*}(x)$, as required by Proposition \ref{main_prop}, is challenging and we consider this as part of our future work.  
\section{Input sharing among multiple Control Barrier Functions}\label{Sec:input_sharing}
Although the barrier function $h_1(x)$ has been shown to render the set $\mathcal{H}$ forward invariant and the functions $h_2(x)$ and $h_3(x)$ have been shown to be valid CBFs in the sense of Definition \ref{Barrier_Defn}, it remains to be shown if they can be enforced together. In other words, the three inputs $\{v_x, v_y, \omega_z\}$ are shared by the three barrier constraints in the Quadratic program \eqref{opt-1} and it remains to be established if \eqref{opt-1} remains feasible in the zone of operation of the Slider. 

We next recall the idea of control sharing CBFs from \cite{Sharing_Journal} and adapt it to the case of relative degree one and generalized to $n$-barrier functions, which is the case under consideration in this work. We refer to this scenario as multi-input sharing among multiple CBFs.
\begin{definition}[Control sharing CBFs \cite{Sharing_Journal}]\label{control_sharing_CBFs}
Consider the control system \eqref{affine} and the CBFs $h_{i}$, $i\in\{1,2, \ldots , n\}$, $n>1$, defined on $\mathcal{X}$. The CBFs $h_{i}(x)$ are said to have the control sharing property, if there exists control $u \in \mathcal{U}$ such that for any $x \in \mathcal{X}$
\begin{equation}
 \dot{h}_i \geq - k_i\alpha_i(h_i(x)),\; i\in\{1,2, \ldots , n\}.
 \label{sharing_condition}
\end{equation}
\end{definition}

In this work, we investigate the scenario where the three control inputs $(v_x,v_y,\omega_z)$ are shared among three barrier functions. The conditions in  \eqref{sharing_condition} require that the control sharing be valid in establishing both forward invariance of the safe sets and asymptotic stability of the safe sets. In general, establishing control sharing in both scenarios is extremely challenging. In this work, we are tasked with the analysis of a relatively simpler scenario for the following reasons : 1) The function $h_1(x)$ has been shown to render the set $\mathcal{H}$ forward invariant but not asymptotically stable. 2) The functions $h_2(x)$ and $h_3(x)$ cannot simultaneously be negative. This is evident from Figure \ref{fig:vis_lock}. In the region outside the red dashed lines, only one of $h_2(x)$ and $h_3(x)$ can be negative. This implies that, effectively, the control sharing needs to be investigated in two cases 1) $h_i(x) \geq 0$ for $i\in\{1,2,3\}$ and 2) $h_1(x) \geq 0$ and $h_2(x)<0$. In the second case, as $h_2(x)<0$ implies $h_3(x) \geq 0$, the latter condition is not explicitly stated.
%
\subsection{Invariance rendering CBFs}\label{Sec:Feasibility_Analysis}
In this subsection, we show that the set $\mathcal{I}=\mathcal{H}\cap\mathcal{V}_2\cap\mathcal{V}_3$ is rendered positively invariant when the three barrier constraints are enforced simultaneously through the Quadratic Program \eqref{opt-1}. 
\begin{proposition}
The set of control inputs $K_{h_1}(x)\cap K_{h_2}(x)\cap K_{h_3}(x)$ as defined in \eqref{control_family} for the functions $h_i(x)$ is non-empty for all $x \in \mathcal{I}$ and any Lipschitz continuous controller $u(x) \in K_{h_1}(x)\cap K_{h_2}(x)\cap K_{h_3}(x)$ renders the set $\mathcal{I}$ forward invariant under the dynamics \eqref{Eq:Slider_reduce}.
\end{proposition}
\begin{proof}
Consider the three conditions  \eqref{barrier_condition_2}, \eqref{barrier_condition_3} and \eqref{barrier_condition_1}. For each $x$, these conditions define half-spaces in the input space $\mathbb{R}^3$ denoted by $H_i(x)$. When $h_2(x) \geq 0$ and $h_3(x) \geq 0$, the planes defining the boundaries of conditions \eqref{barrier_condition_2} and \eqref{barrier_condition_3} are parallel and offset on opposite sides of the origin, with the origin contained in their interior. The set $H_2(x) \cap H_3(x)$ is non-empty and contains the origin and therefore $K_{h_1}(x)\cap K_{h_2}(x)$ is non-empty. When $h_1(x) \geq 0$, the set $H_1(x)$ is a half-space containing the origin in its interior and the plane constituting its boundary is parallel to the $\omega_z$ axis. Therefore, for any $x\in\mathcal{I}$, the intersection of the three half-spaces and $K_{h_1}(x)\cap K_{h_2}(x)\cap K_{h_3}(x)$ are non-empty. The forward invariance of the set $\mathcal{I}$ under any Lipschitz continuous controller $u(x) \in K_{h_1}(x)\cap K_{h_2}(x)\cap K_{h_3}(x)$ is concluded by virtue of Proposition \ref{main_prop}.
\end{proof}
\subsection{Visual lock in finite time}\label{finite_lock}
In this subsection, we begin by showing that the visual locking mode is asymptotically reached, when the system is initialized in the safe zone with no visual lock on the docking port. 
\begin{proposition}\label{asymptotic_lock}
The set  $K_{h_1}(x)\cap K_{h_2}(x)$ as defined in \eqref{control_family} for the functions $h_1(x)$ and $h_2(x)$ is non-empty for all $x \in \mathcal{H}\cap(\mathcal{D}_2\setminus\mathcal{V}_2)$ and any Lipschitz continuous control input $u(x) \in K_{h_1}(x)\cap K_{h_2}(x)$ renders the set $\mathcal{H}$ forward invariant and the set $\mathcal{V}_2$ asymptotically stable for the dynamics \eqref{Eq:Slider_reduce}.
\end{proposition}
\begin{proof}
By Lemma \ref{Lemma2}, $h_2(x)$ is a CBF in the set $\mathcal{D}_2$ and $K_{h_2}(x)$ is non-empty for all $x\in\mathcal{D}_2$. Similarly by Lemma \ref{Lemma1}, $K_{h_1}(x)$ is non-empty for all $x\in\mathcal{H}$. For all $x \in \mathcal{H}\cap(\mathcal{D}_2\setminus\mathcal{V}_2)$, the set $K_{h_1}(x)\cap K_{h_2}(x)$ is non-empty and by Proposition \ref{main_prop}, any Lipschitz continuous controller $u(x)\in K_{h_1}(x)\cap K_{h_2}(x)$ renders the set $\mathcal{H}$ forward invariant and the set $\mathcal{V}_2$ asymptotically stable. 
\end{proof}

For enhanced autonomy, it is desirable that visual locking mode is reached in finite time, so that autonomous operation can be resumed after visual locking. This can be achieved through the following strategy. When the absence of the visual lock is detected, the barrier functions $h_2(x)$ and $h_3(x)$ are redefined with the parameter $0<\bar{\theta}_v \ll \theta_v$. By Proposition \ref{asymptotic_lock}, the visual locking mode, parameterized by $\bar{\theta}_v$, is achieved asymptotically and this implies that the visual locking mode parameterized by $\theta_v$, is reached in finite time, following which the parameter in $h_2(x)$ and $h_3(x)$ can be switched back to $\theta_v$.  

\begin{figure*}[h]
    \centering
\begin{multicols}{2}
\includegraphics[width=0.35\textwidth]{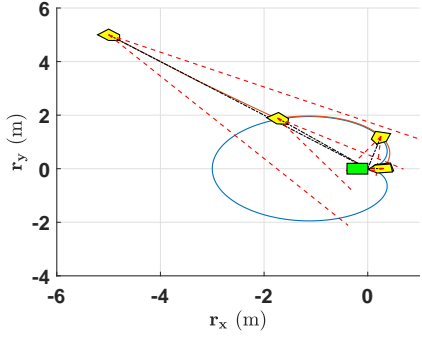}\par\subcaption{The slider achieves safe autonomous docking with a visual lock on the docking port throughout the docking maneuver.\label{sim1p1}}
  \includegraphics[width=0.35\textwidth]{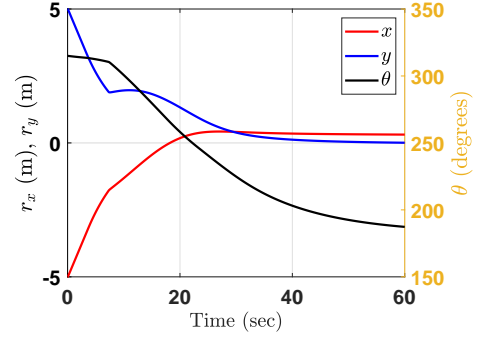}\par\subcaption{The evolution of the states of the slider indicate that the Slider asymptotically reaches the docking port with orientation of $\unit{\pi}{\radian}$. \label{sim1p2}}
    \end{multicols}
 \begin{multicols}{2}   
  \includegraphics[width=0.35\textwidth]{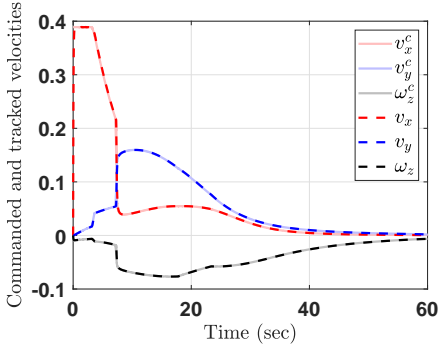}\par\subcaption{The velocities $\{v_x^c, v_y^c, \omega_z^c\}$ commanded by the Quadratic program \eqref{opt-1} are faithfully tracked by the inner-loop controller \eqref{outer-loop}. \label{sim1p3}}
   \includegraphics[width=0.35\textwidth]{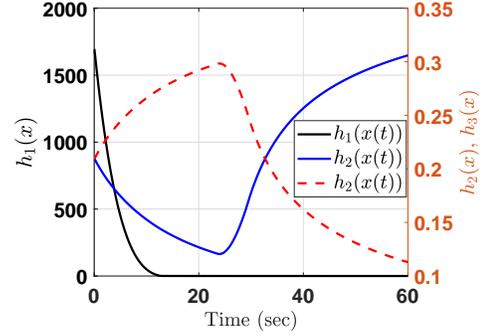}\par\subcaption{The barrier functions $h_i(x)$ are maintained positive throughout the docking maneuver \label{sim1p4}}
  \end{multicols}
  \caption{Simulation results illustrating the invariance of the set defining safety and visual locking. The Slider is initialized in the safe set with a visual lock on the docking port.}
  \label{Stationary_Docking}
\end{figure*}

\section{Simulation results}\label{Sec:Simulation_Results}
In the simulation results presented in this section, the design parameters defining the Cardioid and the range of the visual sensor are $a=\unit{0.75}{\meter}$ and $\theta_v=\unit{\frac{\pi}{15}}{\radian}$ respectively. The parameters of the Slider platform are $m=\unit{4.82}{\kilo\gram}$ and $I_{zz}=\unit{0.11}{\kilo\gram\meter^2}$. 
The admissible control set $\mathcal{U}=[-0.5 \;\; 0.5]^3$ restricts the magnitude of the linear velocities and the angular velocities to within \unit{0.5}{\meter/\second} and \unit{0.5}{\radian/\second} respectively. The parameters $p_i$ in the nominal linear controller from subsection \ref{Quad_prog} are chosen to be 0.15. The parameters $c_i$ in the inner-loop control \eqref{outer-loop} are chosen to be 20. The quadratic program \eqref{sharing_condition} is solved and the control loop is closed at 100 Hz. 
%
%

\subsubsection{Case 1 : Invariance of safety and visual locking for initialization in $\mathcal{H}\cap\mathcal{V}_2\cap\mathcal{V}_3$}
In Figure \ref{Stationary_Docking}, were present simulation results for the case where the docking port is at the origin of the $(x, y)$ plane and orientated along the $x_I$-axis. The docking maneuver begins with the Slider in the safe zone and with a visual lock on the docking port. In Figure \ref{sim1p1}, we see that the Slider reaches the docking port asymptotically while avoiding the area inside the Cardioid and while maintaining a visual lock on the docking port throughout the maneuver. In Figure \ref{sim1p2}, we see that the position of the Slider goes to the origin asymptotically and the Slider reaches the origin pointing in the negative $x$-direction. From Figure \ref{sim1p3}, we see that the kinematic inputs commanded by the Quadratic program \eqref{opt-1} are contained within the set $\mathcal{U}$ and are faithfully tracked by the inner-loop controller \eqref{outer-loop}. From Figure \ref{sim1p4}, we see that the barrier functions $h_i(x)$ remain positive throughout the docking maneuver, indicating the positive invariance of the set $\mathcal{H}\cap\mathcal{V}_2\cap\mathcal{V}_3$.  
\subsubsection{Case 2 : Finite time convergence to visual locking}
In Figure \ref{Attractivity_Docking}, we validate the strategy proposed in \ref{finite_lock} for establishing a visual lock on the docking port in finite time. The Slider is initialized in the safe zone, but with the docking port outside the visual range of the Slider. From Figure \ref{sim2p1}, we see that visual locking is achieved in finite-time and the Slider approaches the docking station asymptotically, navigating through the safe zone. From Figure \ref{sim2p4}, it can be verified that $h_2(x)$ turns positive in finite time and $h_1(x)$ is maintained positive, ensuring safety throughout the docking maneuver. In Figure \ref{sim2p3}, we see that the velocities $\{v_x^c, v_y^c, \omega_z^c\}$ commanded by the Quadratic program \eqref{opt-1} are faithfully tracked by the inner-loop controller \eqref{outer-loop}.
\begin{figure*} 
    \centering
    \begin{multicols}{2}
    \includegraphics[width=0.38\textwidth]{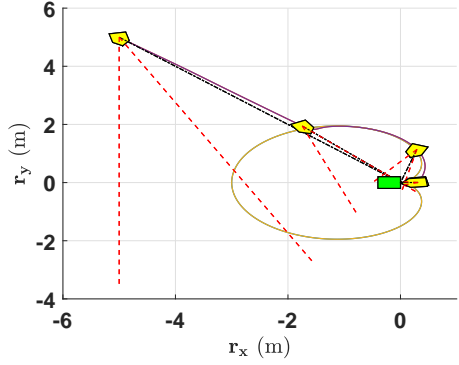}\par\subcaption{The Slider achieves safe docking with finite-time convergence to the visual locking mode. \label{sim2p1}}
     \includegraphics[width=0.35\textwidth]{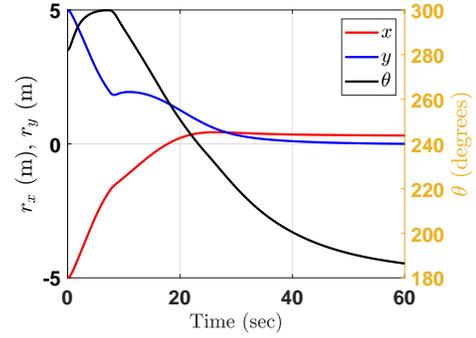}\par\subcaption{The evolution of the states of the slider indicate that the Slider asymptotically reaches the docking port with orientation of $\unit{\pi}{\radian}$. \label{sim2p2}}
    \end{multicols}
 \begin{multicols}{2}
    \includegraphics[width=0.35\textwidth]{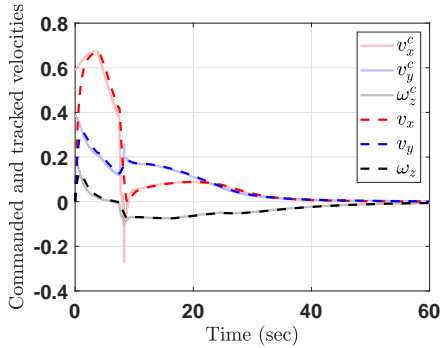}\par\subcaption{The velocities $\{v_x^c, v_y^c, \omega_z^c\}$ commanded by the Quadratic program \eqref{opt-1} are faithfully tracked by the inner-loop controller \eqref{outer-loop}\label{sim2p3}}
     \includegraphics[width=0.35\textwidth]{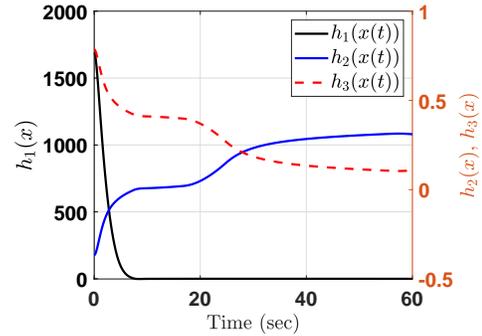}\par\subcaption{The barrier function $h_2(x)$ crosses the $x$-axis into the positive side in finite time, indicating the finite time convergence to the visual locking mode. \label{sim2p4}}
  \end{multicols}
  \caption{Simulation results illustrating the finite time convergence to the visual locking mode when the Slider is initialized in the safe zone but with the docking port out of its visual range.}
  \label{Attractivity_Docking}
\end{figure*}


\section{CONCLUSIONS}\label{Sec:Conclusions}
In this work, we presented a control strategy for safe and autonomous steering of a floating robotic emulation platform, the Slider, to a docking port on a stationary docking station. Control barrier functions (CBFs) were designed to enforce safe distance from the docking station and a correct direction of approach (in one-shot via the Cardioid CBF) and a visual lock on the docking port throughout the docking maneuver. We showed that the barrier functions exhibited the control sharing property in establishing positive invariance of the set defining safety and visual locking. We also presented a control strategy that results in finite-time convergence to the visual locking mode, when from Slider is initialized in the safe set but with no visual lock on the docking port. Simulation results were presented to validate the proposed control strategies. As part of our future work, we look to implement the docking strategies presented in this article on an experimental Slider platform and address the scenario of docking onto a moving station.




\bibliographystyle{ieeetr}
\bibliography{Barrier_Refs}

\begin{thebibliography}{10}

\bibitem{sorgenfrei2014operational}
M.~Sorgenfrei and M.~Nehrenz, ``Operational considerations for a swarm of
  cubesat-class spacecraft,'' in {\em SpaceOps 2014 Conference}, p.~1679, 2014.

\bibitem{elissa}
C.~Trentlage, J.~Yang, M.~B. Larbi, C.~de~Alba~Padilla, and E.~Stoll, ``The
  elissa laboratory: Free-floating satellites for space-related research,'' in
  {\em Deutscher Luft-und Raumfahrtkongress, Friedrichshafen, Germany}, 2018.

\bibitem{testf2}
M.~Sabatini, M.~Farnocchia, and G.~B. Palmerini, ``Design and tests of a
  frictionless 2d platform for studying space navigation and control
  subsystems,'' in {\em 2012 IEEE Aerospace Conference}, pp.~1--12, IEEE, 2012.

\bibitem{Slider_RAL}
A.~Banerjee, S.~G. Satpute, C.~Kanellakis, I.~Tevetzidis, J.~Haluska, P.~Bodin,
  and G.~Nikolakopoulos, ``On the design, modeling and experimental
  verification of a floating satellite platform,'' {\em IEEE Robotics and
  Automation Letters}, vol.~7, no.~2, pp.~1364--1371, 2022.

\bibitem{Docking_Mechanism}
T.~Gasparetto, A.~Banerjee, I.~Tevetzidis, J.~Haluska, C.~Kanellakis, and
  G.~Nikolakopoulos, ``Design of docking mechanism for refueling free-flying 2d
  planar robot,'' in {\em 2021 Aerial Robotic Systems Physically Interacting
  with the Environment (AIRPHARO)}, pp.~1--8, 2021.

\bibitem{barrier_main_ref}
A.~D. Ames, S.~Coogan, M.~Egerstedt, G.~Notomista, K.~Sreenath, and P.~Tabuada,
  ``Control barrier functions: Theory and applications,'' in {\em 2019 18th
  European Control Conference (ECC)}, pp.~3420--3431, 2019.

\bibitem{main_2}
A.~D. Ames, X.~Xu, J.~W. Grizzle, and P.~Tabuada, ``Control barrier function
  based quadratic programs for safety critical systems,'' {\em IEEE
  Transactions on Automatic Control}, vol.~62, no.~8, pp.~3861--3876, 2017.

\bibitem{time_varying_human_assist}
M.~Igarashi, I.~Tezuka, and H.~Nakamura, ``Time-varying control barrier
  function and its application to environment-adaptive human assist control,''
  {\em IFAC-PapersOnLine}, vol.~52, no.~16, pp.~735--740, 2019.
\newblock 11th IFAC Symposium on Nonlinear Control Systems NOLCOS 2019.

\bibitem{time_varying_ttl}
L.~Lindemann and D.~V. Dimarogonas, ``Control barrier functions for signal
  temporal logic tasks,'' {\em IEEE Control Systems Letters}, vol.~3, no.~1,
  pp.~96--101, 2019.

\bibitem{finite-time-tbd}
M.~Srinivasan, S.~Coogan, and M.~Egerstedt, ``Control of multi-agent systems
  with finite time control barrier certificates and temporal logic,'' in {\em
  2018 IEEE Conference on Decision and Control (CDC)}, pp.~1991--1996, 2018.

\bibitem{Sharing_Journal}
X.~Xu, ``Constrained control of input–output linearizable systems using
  control sharing barrier functions,'' {\em Automatica}, vol.~87, pp.~195--201,
  2018.

\bibitem{Sharing_Automation}
A.~Katriniok, ``Control-sharing control barrier functions for intersection
  automation under input constraints,'' 2021.

\bibitem{Multiple-barriers}
W.~Shaw~Cortez, X.~Tan, and D.~V. Dimarogonas, ``A robust, multiple control
  barrier function framework for input constrained systems,'' {\em IEEE Control
  Systems Letters}, vol.~6, pp.~1742--1747, 2022.

\end{thebibliography}


\end{document}